\documentclass{article}
\usepackage{graphicx} 
\usepackage{amsmath}
\usepackage{amsfonts}
\usepackage{natbib}
\usepackage{amsthm}

\usepackage[utf8]{inputenc}
\usepackage[T1]{fontenc}
\usepackage{lmodern}
\usepackage{geometry}
\usepackage{setspace}
\usepackage{graphicx}
\usepackage{amsmath, amssymb, amsthm}
\usepackage{booktabs}
\usepackage{caption}
\usepackage{authblk} 
\usepackage{natbib}
\usepackage{hyperref}
\usepackage{multirow}
\usepackage{subcaption}
\usepackage{algorithm}
\usepackage{algpseudocode}
\usepackage{xcolor}
\usepackage{lineno}
\usepackage{tcolorbox}
\hypersetup{
    colorlinks=true,
    linkcolor=blue,
    citecolor=blue,
    urlcolor=blue
}
\usepackage{doi}
\usepackage[printonlyused]{acronym}
\usepackage{tikz}
\usetikzlibrary{positioning, arrows.meta, decorations.pathreplacing, shapes.geometric, arrows.meta, positioning, calc, shadows.blur, backgrounds, fit}

\newcommand{\Var}{\mathrm{Var}_{\mathbb{P}^*} }
\newcommand{\Cov}{\mathrm{Cov}_{\mathbb{P}^*} }
\newtheorem{prop}{Proposition}
\newtheorem{theorem}{Theorem}
\newtheorem{lemma}{Lemma}
\newtheorem{cor}{Corollary}

\title{Sequential learning theory for Markov genealogy processes}
\author[1]{David J. Pascall}
\affil[1]{MRC Biostatistics Unit, University of Cambridge, Cambridge, UK}
\date{}

\begin{document}

\maketitle

\section{Abstract}

We introduce a filtration-based framework for studying when and why adding taxa improves phylodynamic inference, by constructing a natural ordering of observed tips and applying sequential Bayesian theory to the resulting filtration. The answer turns out to depend crucially on whether the estimand of interest changes structurally on taxa addition. By relating the estimands at intermediate tip numbers to the estimand on the unobserved full genealogy, we generate a classification of estimand types termed learning classes.

\section{Introduction}

A natural question in phylodynamic inference is what happens to inference when taxa are added or removed. Practitioners have observed that taxa addition can cause large changes to tree estimates \citep{collienne2025}, yet the theoretical foundations for understanding when and why taxa addition helps or hurts are largely absent. One intuition as to why taxa addition may not be beneficial in general is that the number of possible trees grows super-exponentially with the number of tips, so unless information grows correspondingly fast, the posterior under a Bayesian model may become increasingly diffuse across the space. However, standard results from sequential Bayesian analysis guarantee that posterior variance about a fixed target cannot increase on average as data accumulates, suggesting that taxa addition should improve precision. The resolution of this tension depends on what is being estimated: the impact of taxa addition on the estimate of a substitution rate under a molecular clock may be very different from the impact on the root age, due to the structural differences in how these estimands relate to the underlying tree.

In this work we develop a framework that allows the exploration of learning dynamics for complex phylodynamic estimands on taxa addition. By constructing a filtration on sequence data via a random ordering of observed tips, we are able to apply standard sequential Bayesian machinery tip addition. We show that the learning process depends critically on how the estimand changes as taxa are added, and that this dependence admits a classification into learning classes, each with distinct theoretical guarantees. Furthermore, by distinguishing between the value of an estimand with respect to the full unobserved genealogy and the analyst's observations of it, we formalise the inferential cost of unobserved lineages, demonstrating that for certain classes of estimands, even when the sequential estimand coincides with the limit target, the analyst's inability to know this fact results in a strictly positive, irreducible gap in posterior variance compared to an oracle. The framework is illustrated throughout using the time to the most recent common ancestor (tMRCA), where we identify a geometric condition on the sample, straddling, that serves as the concrete persistence event governing the estimand's behaviour.

\section{Mathematical setup}
We work on a probability space $(\Omega,\mathcal{F},\mathbb{P})$ supporting a random element $\Delta$.

\[\Delta = (\Theta,\mathcal{G},\Lambda)\]

We define the law $\mathbb{P}$ as follows:
\begin{gather*}
    \Theta \sim \pi\\
    \mathcal{G}|\Theta \sim P_\Theta\\
    h(\mathcal{G}) = \{(x_i,t_i)\}_{i=1}^{f(\mathcal{G})}\\
    \Lambda|f(\mathcal{G}) \sim \mathrm{Uniform}(\mathrm{Sym}(f(\mathcal{G})))
\end{gather*}

Let $\Theta$ be the random variable describing the parameters of an arbitrary Markov genealogy process (MGP) with sampling and mutation, per the definition of \citet{king2025genealogy}, and let $P_\Theta$ be the conditional law of the genealogy. We limit the focus to MGPs that generate genealogies on finite numbers of tips. We have some arbitrary prior $\pi$ on $\Theta$. Let $\mathcal{G}$ be the latent marked tree-valued random variable given by $P_\Theta$. Let $f(\mathcal{G})$ be the deterministic function of $\mathcal{G}$ counting the number of observed samples under the sampling process, and $z(\mathcal{G})$ be the deterministic function counting all the tips on the latent full genealogy, so $f(\mathcal{G})\leq z(\mathcal{G})$. Let $h(\mathcal{G})$ be the deterministic function extracting the set of paired sequences, $x$, and sampling times, $t$, of the observed nodes on the tree. These samples induce an implicit subtree of $\mathcal{G}$. $\Lambda|f(\mathcal{G})$ is uniformly distributed over $\mathrm{Sym}(f(\mathcal{G}))$, such that $\Lambda$ is conditionally independent of $(\Theta,\mathcal{G})$ given $f(\mathcal{G})$.

As the minimal phylogenetically interesting tree is on 3 tips, henceforth, we will work with the new conditional law, $\mathbb{P}^*$, defined by conditioning on $f(\mathcal{G})\geq m+1$ for some fixed integer $m \geq3$.

We use $\Lambda$ to define an ordering of the data, and then work with the natural filtration on this data.
\begin{gather*}
    Y_k=(x_{\Lambda(k)},t_{\Lambda(k)})\\
    D_n = (Y_1,...,Y_n)\\
    \mathcal{F}_n=\sigma(D_n)
\end{gather*}
By construction, $\mathcal F_n\subseteq \mathcal F_{n+1}$. Under $\mathbb P^*$, the filtration is defined for all $n\le m+1$, since $f(\mathcal G)\ge m+1$ a.s.

We then define $\hat{\pi}_n(\cdot)=\mathbb{P}^*((\Theta,\mathcal{G}) \in \cdot|\mathcal{F}_n)$, a deterministic function of $D_n$, so that $\hat{\pi}_n$ is the posterior distribution obtained by conditioning on $D_n$ and using the initial prior, $\pi$. Let $k$ be a measurable function of $(\Theta,\mathcal{G})$, and let $K=k(\Theta,\mathcal{G})$ denote the estimand of interest. The induced posterior on $K$ is given by the pushforward $\hat{\pi}_n^K=\hat{\pi}_n\circ k^{-1}$. The posterior variance of $K$ is then $\mathrm{Var}_{\hat{\pi}_n}(K)$, which equals $\Var(K|\mathcal{F}_n)$ a.s.

It is convenient to define a taxonomy of estimand classes. Permutation-invariant estimands are measurable functions of $(\Theta,\mathcal G)$, and include fixed estimands, $K_f$, and limit estimands, $K_\infty$, as subclasses. Permutation-variant estimands are functions of $(\Theta,\mathcal G,\Lambda)$.
A permutation-invariant estimand $K_\infty$ is a \textit{limit estimand} if there exists a family of hypothetical permutation-variant estimands $(\widetilde K_n(\xi))_{n \le z(\mathcal{G})}$, indexed by permutations $\xi \in \mathrm{Sym}(z(\mathcal{G}))$, such that for $\mathbb{P}$-a.e. $\omega \in \Omega$,
\[
\widetilde K_{z(\mathcal{G})}(\xi)(\omega) = K_\infty(\omega)
\quad \text{for all } \xi \in \mathrm{Sym}(z(\mathcal{G}(\omega))).
\]
with permutation indices are understood pathwise, such that on a realisation $\omega$, $\xi$ ranges over $\mathrm{Sym}(z(\mathcal G(\omega)))$. That is, for every ordering of the $z(\mathcal{G}(\omega))$ observations, the terminal value of the induced sequence coincides with $K_\infty(\omega)$. We refer to these sequences as \textit{hypothetical sequential estimands}.
Equivalently, for each $\xi$, the sequence $(\widetilde K_n(\xi))_{n \le z(\mathcal{G})}$ converges to $K_\infty$ as $n \to z(\mathcal{G})$, though in the finite setting this is simply a statement about the permutation-invariance of the terminal value.
\textit{Fixed estimands} are degenerate limit estimands for which $\widetilde K_n(\xi)=K_f$ for all $n$ and $\xi$. The observed sequential estimand under a realisation of the process is written $(K_n)_{n \le f(\mathcal G)}$.

Note that, for permutation-variant estimands, conditioning on $\mathcal{F}_n$ resolves the $\Lambda$-dependence mediated through the induced subtree a.s., since $D_n$ determines which sampled tips have been revealed among the first $n$ observations, so $K_n$ becomes a function of $(\Theta, \mathcal{G})$ conditional on $D_n$ and the pushforward remains well-defined.

\section{Learning classes}

The first thing that the setup above allows us to do is classify types of estimands into what we call \textit{learning classes}. Learning classes aim for a full classification of sequential estimands in MGPs by relating the permutation-variant estimands in the sequential estimand to their limit estimand, $K_\infty$. Sequential estimands are first divided by whether the mismatch path is monotone. They are then divided by whether equality with the limit estimand is achievable before $z(\mathcal{G})$ with positive probability and if it is, once reached, is persistent in all, some or no cases. Constant sequences are then separated into their own class of degenerate absorbing monotone estimands.

To make the classification precise, we introduce the notion of a persistence event sequence. Let $(\widetilde K_n(\xi))_{n \le z(\mathcal{G})}$ be a hypothetical sequential estimand indexed by $\xi \in \mathrm{Sym}(z(\mathcal{G}))$. We assume that estimands take values in a metric space $(S,d)$, so that the mismatch $d(K_\infty, \widetilde K_n(\xi))$ and the notion of convergence to a unique limit are well-defined. A \textit{persistence event sequence} is a family of events $(\widetilde E_n(\xi))_{n\le z(\mathcal G)}$ such that, for each admissible $\xi$,
\[
\widetilde E_n(\xi)\in \sigma(\Theta,\mathcal G),
\qquad
\widetilde E_n(\xi)\subseteq \widetilde E_{n+1}(\xi),
\]
and
\[
\widetilde E_n(\xi)(\omega)\implies\widetilde K_m(\xi)(\omega)=K_\infty(\omega)
\quad\text{for all }m\ge n
\]
for $\mathbb P$-a.e. $\omega\in \widetilde E_n(\xi)$, where, by slight abuse of notation, $\widetilde E_n(\xi)(\omega)$ denotes the truth value of the event $\widetilde E_n(\xi)$ at $\omega$.

That is, on the event $\widetilde E_n(\xi)$, equality with the limit estimand is attained at time $n$ and persists for all subsequent indices. Every hypothetical sequential estimand admits such a sequence (for example, by taking $\widetilde E_n(\xi) = \{\widetilde K_m(\xi)=K_\infty \ \forall m \ge n\}$). The classification therefore concerns whether persistence occurs with positive probability before $z(\mathcal{G})$, and whether equality with $K_\infty$ can occur without persistence.

When absorption times are discussed below, they are defined in terms of the persistence event sequence. The observed restriction of the persistence event sequence is written $(E_n)_{n \le f(\mathcal G)}$, paralleling the observed sequential estimand $(K_n)_{n \le f(\mathcal G)}$.

We classify sequential estimands according to pathwise properties of their associated hypothetical sequential estimands.

We say that the sequential estimand is \textit{monotonic} if for $\mathbb{P}$-a.e. $\omega$, and for all $\xi \in \mathrm{Sym}(z(\mathcal{G}(\omega)))$, the sequence $(d(K_\infty(\omega), \widetilde K_n(\xi)(\omega)))$ is non-increasing in $n$. Otherwise it is \textit{non-monotonic}.

We say that equality is \textit{achievable} if the set
\[
\{\omega \in \Omega : \exists \xi \in \mathrm{Sym}(z(\mathcal{G}(\omega))),\ \exists n < z(\mathcal{G}(\omega)) \text{ such that } \widetilde K_n(\xi)(\omega) = K_\infty(\omega)\}
\]
has positive $\mathbb{P}$-measure, and \textit{terminal} otherwise.

We say that persistence is \textit{achievable} if the set
\[
\{\omega \in \Omega : \exists \xi \in \mathrm{Sym}(z(\mathcal{G}(\omega))),\ \exists n < z(\mathcal{G}(\omega)) \text{ such that } \widetilde E_n(\xi)(\omega)\}
\]
has positive $\mathbb{P}$-measure. Equality not being achievable implies persistence is not achievable.

The classes are as follows:
\begin{itemize}
    \item \textbf{Fixed:} For $\mathbb{P}$-a.e. $\omega$, $\widetilde K_n(\xi)(\omega) = K_f(\omega)$ for all $\xi$ and $n$.
    \item \textbf{Absorbing monotonic:} Monotonic, non-constant estimands for which persistence is achievable, and for $\mathbb{P}$-a.e. $\omega$, for all $\xi$ and $n$,
    \[
    \widetilde K_n(\xi)(\omega) = K_\infty(\omega) \implies \widetilde E_n(\xi)(\omega).
    \]
    \item \textbf{Absorbing non-monotonic:} Non-monotonic, non-constant estimands for which persistence is achievable, and for $\mathbb{P}$-a.e. $\omega$, for all $\xi$ and $n$,
    \[
    \widetilde K_n(\xi)(\omega) = K_\infty(\omega) \implies \widetilde E_n(\xi)(\omega).
    \]
    \item \textbf{Mixed non-monotonic:} Non-monotonic, non-constant estimands for which persistence is achievable, and the set
    \[
    \{\omega : \exists \xi \in \mathrm{Sym}(z(\mathcal{G}(\omega))), n < z(\mathcal{G}(\omega)) \text{ such that } \widetilde K_n(\xi)(\omega)=K_\infty(\omega) \text{ and } \lnot\widetilde E_n(\xi)(\omega)\}
    \]
    has positive $\mathbb{P}$-measure.
    \item \textbf{Non-absorbing non-monotonic:} Non-monotonic, non-constant estimands for which equality is achievable but persistence is not.
    \item \textbf{Terminal monotonic:} Monotonic,  non-constant estimands for which equality is not achievable.
    \item \textbf{Terminal non-monotonic:} Non-monotonic, non-constant estimands for which equality is not achievable.
\end{itemize}

Mixed monotonic and non-absorbing monotonic classes can be defined in the obvious way, but are definitionally empty, as monotonicity ensures that any equality is persistent.

Some remarks are worth making at this point. Firstly, these definitions are stated under $\mathbb{P}$; they hold under $\mathbb{P}^*$ by restriction, since conditioning on $f(\mathcal{G}) \geq m+1$ does not alter the pathwise properties of the estimand sequences. Secondly, both the definition of sequential estimands and the learning class taxonomy can be extended to the non-finite case, but we leave this for future work. Thirdly, the classification system introduced here assumes only that estimands take values in a metric space $(S,d)$. The metric structure ensures that the notion of monotonic convergence to a unique limit estimand is well-defined. The decomposition results that follow require the stronger assumption that estimands take values in a real Hilbert space, so that subtraction, norms, and covariance (via the inner product) are well defined. The results below do not apply to the tree object directly, though they do apply to some functions of it, because tree spaces are not linear spaces and do not admit an intrinsic Hilbert space structure; there is no natural method of subtracting two trees to get a third tree. Finally, the results below that use the persistence event sequence and absorption time require specification of a concrete persistence event for the estimand in question.

These classes have example estimands that appear in phylodynamic practice. Tree length is terminal monotonic, as long as branch lengths are strictly positive; until every branch is in the tree, which requires every tip to be sampled, the tree length does not attain its latent value, but adding a tip always adds a branch with a positive length, so the estimand value increases monotonically with tip addition. Different measures of tree imbalance fall into different classes. Colless' index \citep{colless1982} at the root is non-absorbing non-monotonic, as it changes in an unstructured manner, but the difference can equal the limit difference multiple times as tips are added. Colless' index at a given non-root internal node is mixed non-monotonic, as equality with the limit difference can be reached multiple times as tips are added, and there exists a persistence event when all the taxa from the subtree subtended by the focal node have been sampled. The clock rate is fixed; adding a tip may change the estimated value, but doesn't structurally change what is being estimated. tMRCA-type estimands are absorbing monotonic; the monotonicity is easy to see, adding a tip can only push the tMRCA into the past or not change it, to see that it is absorbing we need to make concrete the persistence event. The relevant event is what we call \textit{straddling}. Consider the root of the latent genealogy, a set of tips straddles this root if the set contains a tip from at least two distinct subtrees descending from it. As soon as a set of tips straddles the root, the tMRCA of that set of tips is fixed to the root age, the earliest it can possibly be. Straddling for non-root nodes only implies that the tMRCA of that set is at least as old as the age of that node, but it is easy to introduce a notion of \textit{strict straddling}, which adds the condition that the set does not contain any tips that don't descend from that node, which carries the properties at the root through to non-root nodes. As straddling can be attained as early as two tips into an ordering, it meets the criteria to be a persistence event.

Identifying the class of an estimand in general is not simple, as it requires both a structural understanding of the mismatch path between the sequence and the limit for all possible permutations, identification of the relevant persistence event, and assessment of whether that event is possible under the process. However, this identification process allows the identification of certain learning properties, some of which are proven below.

\section{Learning for sequential estimands}

We can use the fact that standard learning applies to limit estimands in order to say something about learning for general phylodynamic sequential estimands, under the assumption that the target we actually care about is the limit target. From this point onward, unless otherwise stated, estimands are assumed to take values in a real Hilbert space so that conditional variances and covariances are well-defined, with $\Var(K|\mathcal{F}_n)$ being a.s. equal to $\mathbb E_{\hat\pi_n}[\|K-\mathbb E_{\hat\pi_n}[K]\|^2]$ and $\Cov(K,|\mathcal{F}_n)$ being a.s. equal to $\mathbb E_{\hat\pi_n}[
\langle K-\mathbb E_{\hat\pi_n}[K],L-\mathbb E_{\hat\pi_n}[L]\rangle
]$.

\begin{prop}[Standard learning for permutation-invariant estimands]
\label{standardlearning}
Under the construction above, $(\mathcal{F}_n)_{n\leq f(\mathcal{G})}$ is a filtration on $(\Omega,\mathcal{F},\mathbb{P}^*)$, and for any square-integrable permutation-invariant estimand, $K$, $\Var(K|\mathcal{F}_n)\geq \mathbb{E}_{\mathbb{P}^*}[\Var(K|\mathcal{F}_{n+1})|\mathcal{F}_n]$ a.s.
\end{prop}
\begin{proof}
    This is a standard result, e.g. see \cite{durrett2010prob}. We simply verify the hypotheses; by construction $\mathcal{F}_n \subseteq \mathcal{F}_{n+1}$, by permutation-invariance the estimand is a fixed target across the filtration (i.e. $K$ does not change with $n$), and by assumption $K$ is square-integrable, therefore the result follows by application of the Law of Total Variance. 
\end{proof}

This result says that if the target remains fixed on taxa addition, then the analyst's current posterior variance provides an upper bound on the expected posterior variance after the next observation, conditional on the current information. A more general version of this result is presented in Appendix 1. There, for estimands taking values in a separable psuedometric spaces, conditional variance is replaced by conditional Fréchet variance. This yields an analogous supermartingale learning result for fixed separable-psuedometric-space-valued estimands, including tree-valued estimands.
 
 We can get better understanding of the learning process for sequential estimands by decomposing the variance of the sequential estimand in terms of the limit estimand.

\begin{lemma}[Variance decomposition for sequential estimands]
\label{vardecomp}
For any estimand, $K_n$, part of a sequential estimand, $(K_n)_{n\leq f(\mathcal{G})}$, with square-integrable limit estimand, $K_\infty$, and square-integrable difference, $K_\infty-K_n$, then,
\begin{align*}
\Var(K_\infty|\mathcal{F}_n)=&
\Var(K_n|\mathcal{F}_n)+\\
&\Var(K_\infty-K_n|\mathcal{F}_n)+\\
&2\Cov(K_n,K_\infty-K_n|\mathcal{F}_n)
\end{align*}
a.s.
\end{lemma}
\begin{proof}
    Note that $K_\infty=K_n+(K_\infty-K_n)$, substitute that into $\Var(K_\infty|\mathcal{F}_n)$, and apply the conditional version of Bienaymé's identity. 
\end{proof}

That is, we can decompose our uncertainty about our limit target into our uncertainty about the current sequential target, our uncertainty about the discrepancy between our current sequential target and the limit target, and how the two uncertainties relate to each other. 

\begin{theorem}[Conditional sequential learning decomposition]
\label{sequentiallearningdecomp}
For any pair of estimands, $K_n$ and $K_{n+1}$, part of a sequential estimand, $(K_n)_{n\leq f(\mathcal{G})}$, with square-integrable limit estimand, $K_\infty$, and square-integrable differences, $K_\infty-K_n$ and $K_\infty-K_{n+1}$, then,
\begin{align*}
    0\leq&\Var(K_{n}|\mathcal{F}_{n})-\mathbb{E}_{\mathbb{P}^*}[\Var(K_{n+1}|\mathcal{F}_{n+1})|\mathcal{F}_{n}]+\\
    &\Var(K_\infty-K_{n}|\mathcal{F}_{n})-\mathbb{E}_{\mathbb{P}^*}[\Var(K_\infty-K_{n+1}|\mathcal{F}_{n+1})|\mathcal{F}_{n}]+\\
    &2(\Cov(K_n,K_\infty-K_n|\mathcal{F}_n)-\mathbb{E}_{\mathbb{P}^*}[\Cov(K_{n+1},K_\infty-K_{n+1}|\mathcal{F}_{n+1})|\mathcal{F}_{n}])
\end{align*} a.s.
\end{theorem}
\begin{proof}
    Apply Lemma \ref{vardecomp} to $\Var(K_\infty|\mathcal{F}_n)$ and $\Var(K_\infty|\mathcal{F}_{n+1})$, then substitute the decompositions into Proposition \ref{standardlearning} and rearrange.
\end{proof}

This relates the previous result to learning, and applies generally across all the learning classes. We have already seen that, conditional on the current information, the expected posterior uncertainty about the limit target after the next observation does not exceed the current posterior uncertainty. This theorem gives a mechanistic breakdown of this learning process. The first pair of terms is analogous to a classical learning term: it is the current uncertainty about the sequential target minus its expected posterior uncertainty after the next observation, conditional on $\mathcal{F}_n$. The second pair is a "mismatch" term, measuring the expected change in our uncertainty about how far the current target is from the limit estimand. The final term is about the expected change in how the uncertainty about the current target and the uncertainty about the discrepancy covary. 

None of the three terms have an obvious sign. The first may be positive or negative, as adding a taxon can increase or decrease the analyst's uncertainty about the current sequential target. The second depends on how the analyst's uncertainty about the mismatch $K_\infty - K_n$ evolves; although the learning classes describe the pathwise behaviour of this mismatch, they constrain but do not determine the analyst's posterior uncertainty about it, since the analyst does not observe the mismatch directly. The third depends on the covariance structure of phylogenetic posteriors, which is complex and poorly understood. Despite this, the theorem guarantees that while these three terms can trade off against each other, their sum must remain non-negative by the improvement guarantees of Proposition \ref{standardlearning}.

For absorbing monotonic, absorbing non-monotone and mixed non-monotonic learning class estimands, we can get a further decomposition in terms of the persistence event, the event after which equality with the limit estimand is maintained a.s. This allows further characterisation of learning for these classes.

\begin{lemma}[Partitioned decomposition conditional on absorption status]
\label{absorptionlearning}
Let $(K_n)_{n\leq f(\mathcal{G})}$ be an absorbing monotonic, absorbing non-monotonic or mixed non-monotonic sequential estimand with square-integrable limit estimand, $K_\infty$, and square-integrable differences, $K_\infty-K_n$, for $n\leq f(\mathcal{G})$. Let $\tau$ be the random time at which the persistence event occurs in the observed persistence event sequence, if it does occur, that is $\tau=\inf\{n\in\{1,...,f(\mathcal{G})\}:E_n\}$, with $\tau=\infty$ if the set is empty, then, 
\begin{align*}
    0\leq&\mathbb{P}^*(\tau \leq n)\mathbb{E}_{\mathbb{P}^*}[\Var(K_\infty|\mathcal{F}_{n})-\Var(K_\infty|\mathcal{F}_{n+1})|\tau \leq n]+\\
    &\mathbb{P}^*(\tau=n+1)\mathbb{E}_{\mathbb{P}^*}[\Var(K_{n}|\mathcal{F}_{n})-\Var(K_\infty|\mathcal{F}_{n+1})|\tau=n+1]+\\
    &\mathbb{P}^*(\tau=n+1)\mathbb{E}_{\mathbb{P}^*}[\Var(K_\infty-K_{n}|\mathcal{F}_{n})|\tau=n+1]+\\
    &2\mathbb{P}^*(\tau=n+1)\mathbb{E}_{\mathbb{P}^*}[\Cov(K_n,K_\infty-K_n|\mathcal{F}_n)|\tau=n+1]+\\    
    &\mathbb{P}^*(n+1<\tau)\mathbb{E}_{\mathbb{P}^*}[\Var(K_{n}|\mathcal{F}_{n})-\Var(K_{n+1}|\mathcal{F}_{n+1})|n+1<\tau]+\\
    &\mathbb{P}^*(n+1<\tau)\mathbb{E}_{\mathbb{P}^*}[\Var(K_\infty-K_{n}|\mathcal{F}_{n})-\Var(K_\infty-K_{n+1}|\mathcal{F}_{n+1})|n+1<\tau]+\\
    &2\mathbb{P}^*(n+1<\tau)\mathbb{E}_{\mathbb{P}^*}[\Cov(K_n,K_\infty-K_n|\mathcal{F}_n)-\Cov(K_{n+1},K_\infty-K_{n+1}|\mathcal{F}_{n+1})|n+1<\tau]
\end{align*}
\end{lemma}
\begin{proof}
    Take expectations of the result of Theorem \ref{sequentiallearningdecomp}, then partition by the events $\{\tau \leq n\}$, $\{\tau = n+1\}$ and $\{n+1<\tau\}$, using the Law of Total Expectation on the total random variable, which is valid because $\tau$ is measurable with respect to $\mathcal{F}$. Then simplify using the facts that on $\{\tau \leq n\}$, $K_n=K_{n+1}=K_\infty$ a.s., and on $\{\tau = n+1\}$, $K_{n+1}=K_\infty$ a.s., so the relevant mismatch and covariance terms vanish.
\end{proof}

This decomposition splits the learning regime into three discrete structural states for absorbing monotonic, absorbing non-monotonic and mixed non-monotonic sequential estimands: When $\tau \leq n$ the mismatch has already vanished and learning is classical. When $\tau = n+1$ the addition of the next taxon resolves the remaining mismatch. When $\tau > n+1$ permanent equality has not been reached (and for absorbing monotonic and absorbing non-monotonic sequential estimands mismatch is guaranteed non-zero) on taxon addition and the structure remains as it was in Theorem \ref{sequentiallearningdecomp}. There are two key points that are relevant to interpreting this result. Firstly, this remains a result averaged across the entire process. It does not say that within every structural state there is variance reduction on taxon addition in expectation, just that this decomposition is possible and the average result of Proposition \ref{standardlearning} applies. Secondly, the analyst cannot make use of this decomposition. What the analyst knows at step $n$ is contained in $\mathcal{F}_n$, but $\tau$ is not contained in $\mathcal{F}_n$ (or $\mathcal{F}_{n+1}$); it is not $\mathcal{F}_n$-measurable.

At this point we can introduce an oracle who knows what the analyst knows plus $\tau$. We may view their expanded filtration is $\sigma(D_n, \tau)=\mathcal{F}_n'$ as an initial enlargement of the analyst’s filtration by the random variable $\tau$, in the sense of the classical enlargement-of-filtrations literature \citep{grigorian2023}. This has immediate consequences.

\begin{cor}[Event-wise in expectation variance reduction for the oracle]
\label{oraclelearning}
Let $(K_n)_{n\leq f(\mathcal{G})}$ be an absorbing monotonic, absorbing non-monotonic or mixed non-monotonic sequential estimand with square-integrable limit estimand, $K_\infty$, and square-integrable differences, $K_\infty-K_n$, for $n\leq f(\mathcal{G})$, and let $\tau$ be the random time at which absorption occurs, as defined in Lemma \ref{absorptionlearning}, then under the oracle filtration $\sigma(D_n, \tau)=\mathcal{F}_n'$, 
\begin{align*}
\mathbb{E}_{\mathbb{P}^*}[\Var(K_n|\mathcal{F}'_n)|\tau \leq n]&\geq \mathbb{E}_{\mathbb{P}^*}[\Var(K_{n+1}|\mathcal{F}'_{n+1})|\tau \leq n]\\
\mathbb{E}_{\mathbb{P}^*}[\Var(K_\infty|\mathcal{F}'_n)|\tau = n+1]&\geq \mathbb{E}_{\mathbb{P}^*}[\Var(K_{n+1}|\mathcal{F}'_{n+1})|\tau = n+1]\\
\mathbb{E}_{\mathbb{P}^*}[\Var(K_\infty|\mathcal{F}'_n)|n+1<\tau]&\geq \mathbb{E}_{\mathbb{P}^*}[\Var(K_\infty|\mathcal{F}'_{n+1})|n+1<\tau]
\end{align*}
\end{cor}
\begin{proof}
    As $\tau$ is $\mathcal{F}'_n$-measurable, Proposition \ref{standardlearning} applies under the oracle filtration. Conditioning the resulting inequality on the events $\{\tau \leq n\}$, $\{\tau = n+1\}$ and $\{n+1<\tau\}$, and simplifying using that on $\{\tau \leq n\}$, $K_n=K_{n+1}=K_\infty$ a.s., while on $\{\tau = n+1\}$, $K_{n+1}=K_\infty$ a.s., gives the result.
\end{proof}

There are two important takeaways from Corollary \ref{oraclelearning}. The first is that the oracle gets guarantees about classical learning for the sequential estimands that the analyst uses in practice, because the oracle knows that they're already the target that they care about. The analyst has to carry a burden of not knowing, which leads to them having to consider the mismatch and covariance terms when the oracle can ignore them. The second is that the oracle gets guaranteed learning in expectation eventwise. As we will see, the analyst has to pay for those same guarantees.

\begin{theorem}[Event-conditional learning]
\label{eventconditionallearning}
Let $K$ be a square-integrable permutation-invariant estimand, and let $E\in \mathcal{F}$ with $\mathbb{P}^*(E)>0$ be some event that is being conditioned on. Let $\mathcal{F}_n\subseteq\mathcal{F}_{n+1}$ be some nested $\sigma$-algebras within $\mathcal{F}$. Write $V_k = \Var(K|\mathcal{F}_k)$ and $Z_k = \mathbb{P}^*(E|\mathcal{F}_k)$, then,
\[ \mathbb{E}_{\mathbb{P}^*}[V_n-V_{n+1}|E]=\mathbb{E}_{\mathbb{P}^*}[V_n-V_{n+1}]+\frac{\Cov(V_n,Z_n)-\Cov(V_{n+1},Z_{n+1})}{\mathbb{P}^*(E)}\]
where the covariance terms are unconditional covariances under $\mathbb P^*$.
\end{theorem}
\begin{proof}
    By definition, $\mathbb{E}_{\mathbb{P}^*}[V_n - V_{n+1}|E] = \frac{\mathbb{E}_{\mathbb{P}^*}[(V_n - V_{n+1})\mathbf{1}_E]}{\mathbb{P}^*(E)}$. For each term, apply the tower property through the appropriate $\sigma$-algebra: since $V_n$ is $\mathcal{F}_n$-measurable, $\mathbb{E}_{\mathbb{P}^*}[V_n \mathbf{1}_E] = \mathbb{E}_{\mathbb{P}^*}[V_n Z_n]$, and since $V_{n+1}$ is $\mathcal{F}_{n+1}$-measurable, $\mathbb{E}_{\mathbb{P}^*}[V_{n+1} \mathbf{1}_E] = \mathbb{E}_{\mathbb{P}^*}[V_{n+1} Z_{n+1}]$. Expanding each product of expectations using $\mathbb{E}[XY] = \mathbb{E}[X]\mathbb{E}[Y] + \Cov(X,Y)$ and noting that $Z_k$ is a martingale with $\mathbb{E}_{\mathbb{P}^*}[Z_k] = \mathbb{P}^*(E)$ for all $k$, the result follows after dividing through by $\mathbb{P}^*(E)$.
\end{proof}

This allows us to directly compare the learning of the oracle and the analyst. For the oracle, who knows the status of the event, their posterior probability of the event is either 1 or 0, so the second term vanishes, and they're left with the term that by Proposition \ref{standardlearning} gives classical learning guarantees. The analyst, for whom the event is uncertain (i.e. their posterior probability of the event lies between 0 and 1), keeps the second term, and fails to learn eventwise in expectation if a negative covariance correction term between their estimand and the posterior probability is sufficient to overwhelm the process guarantees of Proposition \ref{standardlearning}.

\begin{lemma}[Analyst-oracle decomposition]
\label{analytoracledecomp}
Let $(K_n)_{n\leq f(\mathcal{G})}$ be an absorbing monotonic, absorbing non-monotonic or mixed non-monotonic sequential estimand with square-integrable limit estimand, $K_\infty$, and let $\tau$ be the random time at which absorption occurs. Let $\Var(K_\infty|\mathcal{F}_n)$ be the analyst's posterior variance and $\Var(K_\infty|\mathcal{F}'_n)$ be the oracle's posterior variance under $\mathcal{F}'_n=\sigma(D_n,\tau)$. Then,
\[\Var(K_\infty|\mathcal{F}_n)= \mathbb{E}_{\mathbb{P}^*}[\Var(K_\infty|\mathcal{F}'_n)|\mathcal{F}_n]+\Var(\mathbb{E}_{\mathbb{P}^*}[K_\infty|\mathcal{F}'_n]|\mathcal{F}_n)\] 
a.s.
\end{lemma}
\begin{proof}
    This is the Law of Total Variance applied to nested $\sigma$-algebras. As $\mathcal{F}_n \subset \mathcal{F}'_n$, the result follows.
\end{proof}

The second term here, the conditional variance of the oracle's posterior mean given only the analyst's information, is the component of the analyst's uncertainty that is driven by blindness to the absorption status, a direct expression of the cost of the analyst not knowing what the oracle knows.

\begin{theorem}[Irreducibility of the oracle gap]
\label{irreduciblegap}
Let $(K_n)_{n\leq f(\mathcal{G})}$ be a scalar absorbing monotonic sequential estimand with square-integrable limit estimand, $K_\infty$, and square-integrable differences of constant sign, $K_\infty-K_n$. Let $\tau$ be the random time at which absorption occurs, as defined in Lemma \ref{absorptionlearning}, with $0<\mathbb{P}^*(\tau<\infty|\mathcal{F}_{f(\mathcal{G})})<1$ on a set of positive probability. If $\mathbb{E}_{\mathbb{P}^*}[K_{f(\mathcal{G})}|\mathcal{F}_{f(\mathcal{G})}, \tau<\infty]=\mathbb{E}_{\mathbb{P}^*}[K_{f(\mathcal{G})}|\mathcal{F}_{f(\mathcal{G})}, \tau=\infty]$ on the set where $0<\mathbb{P}^*(\tau<\infty|\mathcal{F}_{f(\mathcal{G})})<1$, then, $\mathbb{E}_{\mathbb{P}^*}[\Var(\mathbb{E}_{\mathbb{P}^*}[K_\infty|\mathcal{F}'_{f(\mathcal{G})}]|\mathcal{F}_{f(\mathcal{G})})] > 0$.
\end{theorem}
\begin{proof}
    Write $A=\{\tau<\infty\}$ and $\mathcal{H}=\sigma(\mathcal{F}_{f(\mathcal{G})},A)$. Since $\mathcal{F}_{f(\mathcal{G})}\subset\mathcal{H}\subset\mathcal{F}_{f(\mathcal{G})}'$, the conditional Law of Total Variance gives $\Var(\mathbb{E}_{\mathbb{P}^*}[K_\infty|\mathcal{F}'_{f(\mathcal{G})}]|\mathcal{F}_{f(\mathcal{G})}) \geq \Var(\mathbb{E}_{\mathbb{P}^*}[\mathbb{E}_{\mathbb{P}^*}[K_\infty|\mathcal{F}'_{f(\mathcal{G})}]|\mathcal{H}]|\mathcal{F}_{f(\mathcal{G})})$. As $\mathcal{H}$ is a binary enlargement, the conditional variance identity for binary enlargements gives $\Var(\mathbb{E}_{\mathbb{P}^*}[\mathbb{E}_{\mathbb{P}^*}[K_\infty|\mathcal{F}'_{f(\mathcal{G})}]|\mathcal{H}]|\mathcal{F}_{f(\mathcal{G})})=\mathbb{P}^*(A|\mathcal{F}_{f(\mathcal{G})})\mathbb{P}^*(A^c|\mathcal{F}_{f(\mathcal{G})})(\mathbb{E}_{\mathbb{P}^*}[\mathbb{E}_{\mathbb{P}^*}[K_\infty|\mathcal{F}_{f(\mathcal{G})}']|\mathcal{F}_{f(\mathcal{G})},A]-\mathbb{E}_{\mathbb{P}^*}[\mathbb{E}_{\mathbb{P}^*}[K_\infty|\mathcal{F}_{f(\mathcal{G})}']|\mathcal{F}_{f(\mathcal{G})},A^c])^2$. By the Tower Property, $\mathbb{E}_{\mathbb{P}^*}[\mathbb{E}_{\mathbb{P}^*}[K_\infty|\mathcal{F}_{f(\mathcal{G})}']|\mathcal{F}_{f(\mathcal{G})},A]=\mathbb{E}_{\mathbb{P}^*}[K_\infty|\mathcal{F}_{f(\mathcal{G})},A]$ as $\sigma(\mathcal{F}_{f(\mathcal{G})},A)\subset\mathcal{F}_{f(\mathcal{G})}'$, and likewise on $A^c$. Since $K_\infty=K_{f(\mathcal{G})}+(K_\infty-K_{f(\mathcal{G})})$ and $K_\infty-K_{f(\mathcal{G})}=0$ a.s. on $A$, we have, $\mathbb{E}_{\mathbb{P}^*}[K_\infty|\mathcal{F}_{f(\mathcal{G})},A]-\mathbb{E}_{\mathbb{P}^*}[K_\infty|\mathcal{F}_{f(\mathcal{G})},A^c]=\mathbb{E}_{\mathbb{P}^*}[K_{f(\mathcal{G})}|\mathcal{F}_{f(\mathcal{G})},A]-\mathbb{E}_{\mathbb{P}^*}[K_{f(\mathcal{G})}|\mathcal{F}_{f(\mathcal{G})},A^c]-\mathbb{E}_{\mathbb{P}^*}[K_\infty-K_{f(\mathcal{G})}|\mathcal{F}_{f(\mathcal{G})},A^c]$. We can then cancel using the $\mathbb{E}_{\mathbb{P}^*}[K_{f(\mathcal{G})}|\mathcal{F}_{f(\mathcal{G})}, \tau<\infty]=\mathbb{E}_{\mathbb{P}^*}[K_{f(\mathcal{G})}|\mathcal{F}_{f(\mathcal{G})}, \tau=\infty]$ assumption. $K_\infty\neq K_{f(\mathcal{G})}$ a.s. on $A^c$ by the definition of the event, and by the constant sign assumption the difference has constant sign, so the difference has strict constant sign a.s. on $A^c$. Therefore, $\mathbb{E}_{\mathbb{P}^*}[K_\infty-K_{f(\mathcal{G})}|\mathcal{F}_{f(\mathcal{G})},A^c]\neq0$ whenever $\mathbb{P}^*(A^c|\mathcal{F}_{f(\mathcal{G})})>0$. Hence the squared term above is strictly positive on the set where $0<\mathbb{P}^*(A|\mathcal{F}_{f(\mathcal{G})})<1$, which has positive probability by assumption. Taking expectations with respect to the process gives the required result.
\end{proof} 

That is, even after observing all sampled tips, the analyst's posterior variance strictly exceeds the expected oracle posterior variance, assuming sequential mean-independence from absorption status. Sequential mean-independence from absorption status, $\mathbb{E}_{\mathbb{P}^*}[K_{f(\mathcal{G})}|\mathcal{F}_{f(\mathcal{G})}, \tau<\infty]=\mathbb{E}_{\mathbb{P}^*}[K_{f(\mathcal{G})}|\mathcal{F}_{f(\mathcal{G})}, \tau=\infty]$, encodes the assumption that knowing whether your current sequential target already equals the limit target shouldn't impact your expected estimate of that current target. 

This is very natural assumption for many phylogenetic estimands. Consider the tMRCA of all the tips in the current sample. Knowing whether the current sample has straddled the root should not change the estimate of the tMRCA of the current sample set, but it should still impact your estimate of the age of the root of the full latent genealogy, the sequential mean-independence from absorption status assumption encodes precisely this intuition.

\section{Acknowledgements}
This research was supported by the National Institute for Health and Care Research (NIH) Cambridge Biomedical Research Centre (NIHR203312). The views expressed are those of the authors and not necessarily those of the NIHR or the Department of Health and Social Care. We would like to thank the attendees of the Neural Inference for Bayesian Phylodynamics and Genetic Epidemiology workshop supported by the EPSRC Probabilistic AI Hub (EP/Y028783/1) for helping to find examples of phylogenetically relevant estimands for the different learning classes.

\section{Appendix 1}
\begin{prop}[Standard learning for fixed metric-space-valued estimands]
\label{metricstandardlearning}
Let $(S,d)$ be a separable metric space and let $K$ be a permutation-invariant $S$-valued estimand satisfying $\mathbb{E}_{\mathbb{P}^*}[d^2(K,s_0)]<\infty$ for some $s_0 \in S$. Then the conditional Fr\'{e}chet variance of $K$ given $\mathcal{F}_n$, $\inf_{s \in S} \mathbb{E}_{\mathbb{P}^*}[d^2(K,s)|\mathcal{F}_n]$, is a supermartingale. That is, for all $n < f(\mathcal{G})$, $\mathbb{E}_{\mathbb{P}^*}[\inf_{s \in S} \mathbb{E}_{\mathbb{P}^*}[d^2(K,s)|\mathcal{F}_{n+1}]|\mathcal{F}_n] \leq \inf_{s \in S} \mathbb{E}_{\mathbb{P}^*}[d^2(K,s)|\mathcal{F}_n]$ a.s.
\end{prop}
\begin{proof}
\textit{Well-posedness.} For each $s \in S$, the map $x \mapsto d^2(x,s)$ is continuous, hence $d^2(K,s)$ is $\mathcal{F}$-measurable. Moreover, for any $s \in S$, $d^2(K,s) \leq 2d^2(K,s_0) + 2d^2(s_0,s)$, so $d^2(K,s) \in L^1$ by the basepointed second-moment assumption. Thus $\mathbb{E}_{\mathbb{P}^*}[d^2(K,s)|\mathcal{F}_n]$ is a finite $\mathcal{F}_n$-measurable random variable for each $s$.

For $s, s' \in S$, $|d^2(K,s) - d^2(K,s')| \leq d(s,s')(2d(K,s_0) + d(s_0,s) + d(s_0,s'))$. Since $d(K,s_0) \in L^1$ by Jensen and the basepointing assumption, taking conditional expectations yields $|\mathbb{E}_{\mathbb{P}^*}[d^2(K,s)|\mathcal{F}_n] - \mathbb{E}_{\mathbb{P}^*}[d^2(K,s')|\mathcal{F}_n]| \leq d(s,s')(2\mathbb{E}_{\mathbb{P}^*}[d(K,s_0)|\mathcal{F}_n] + d(s_0,s) + d(s_0,s')) \quad \text{a.s.}$ for each fixed pair $(s,s')$. Fix a countable dense subset $\{s_k\} \subset S$. Restricting to pairs from $\{s_k\}$, the bound holds simultaneously for all pairs outside a single null set; the resulting map extends uniquely to a continuous version of $s \mapsto \mathbb{E}_{\mathbb{P}^*}[d^2(K,s)|\mathcal{F}_n](\omega)$ on $S$. Working with this version, $\inf_{s \in S} \mathbb{E}_{\mathbb{P}^*}[d^2(K,s)|\mathcal{F}_n] = \inf_k \mathbb{E}_{\mathbb{P}^*}[d^2(K,s_k)|\mathcal{F}_n]$ a.s., so the conditional Fr\'{e}chet variance is $\mathcal{F}_n$-measurable as a countable infimum of $\mathcal{F}_n$-measurable random variables. Finally, $0 \leq \inf_{s \in S} \mathbb{E}_{\mathbb{P}^*}[d^2(K,s)|\mathcal{F}_n] \leq \mathbb{E}_{\mathbb{P}^*}[d^2(K,s_0)|\mathcal{F}_n] < \infty$ a.s., implying $\inf_{s \in S} \mathbb{E}_{\mathbb{P}^*}[d^2(K,s)|\mathcal{F}_n] \in L^1$.

\textit{Inequality.} For any fixed $s \in S$, the tower property gives $\mathbb{E}_{\mathbb{P}^*}[d^2(K,s)|\mathcal{F}_n] = \mathbb{E}_{\mathbb{P}^*}[\mathbb{E}_{\mathbb{P}^*}[d^2(K,s)|\mathcal{F}_{n+1}]|\mathcal{F}_n] \geq \mathbb{E}_{\mathbb{P}^*}[\inf_{s' \in S}\mathbb{E}_{\mathbb{P}^*}[d^2(K,s')|\mathcal{F}_{n+1}]|\mathcal{F}_n]$, where the inequality is pointwise: for each $\omega$, the conditional expectation at $s$ is at least the infimum over $s'$. As the right-hand side is independent of $s$, taking the infimum over $s$ on the left gives the result.
\end{proof}

Note that in Hadamard spaces, such as BHV space \citep{BHV2001} or $\tau$-space  \citep{tauspace2016}, this minimiser is unique \citep{Sturm2003ProbabilityMO}, but this result holds in spaces without enough structure to guarantee the existence or uniqueness of a minimiser.

\begin{cor}[Extension to separable pseudometric spaces]
\label{pseudometriclearning}
Proposition \ref{metricstandardlearning} holds when $(S,d)$ is a separable pseudometric space.
\end{cor}
\begin{proof}
Define an equivalence relation on $S$ by $x \sim y \iff d(x,y)=0$, and let $q:S\to S/{\sim}$ be the quotient map. The induced function, $\bar d(q(x),q(y))=d(x,y)$, is a well-defined metric on $S/{\sim}$. Since $S$ is separable, so is $S/{\sim}$, as if $\{s_k\}$ is a countable dense subset of $S$, then $\{q(s_k)\}$ is dense in $S/{\sim}$.

For any $S$-valued estimand $K$ and any $s\in S$, $d(K,s)=\bar d(q\circ K,q(s))$ a.s., so $\mathbb{E}_{\mathbb{P}^*}[d^2(K,s)|\mathcal{F}_n]= \mathbb{E}_{\mathbb{P}^*}[\bar d^2(q\circ K,q(s))|\mathcal{F}_n]$, with $\mathbb{E}_{\mathbb{P}^*}[\bar d^2(q\circ K,q(s))|\mathcal{F}_n]< \infty$ by the transfer of the basepointed second moment condition. Since every element of $S/{\sim}$ is of the form $q(s)$ for some $s\in S$, it follows that $\inf_{s\in S}
\mathbb{E}_{\mathbb{P}^*}[d^2(K,s)|\mathcal{F}_n] = \inf_{z\in S/{\sim}} \mathbb{E}_{\mathbb{P}^*}[\bar d^2(q\circ K,z)|\mathcal{F}_n]$.
Thus the conditional Fr\'{e}chet variance of $K$ in the pseudometric space $(S,d)$ coincides with the conditional Fr\'{e}chet variance of $q\circ K$ in the separable metric space $(S/{\sim},\bar d)$. Applying Proposition \ref{metricstandardlearning} to $(S/{\sim},\bar d)$ gives the result.
\end{proof}

This applies, in particular, to the pseudometric on the Towering space of trees \citep{cabrera2025geometryspacephylogenetictrees}.

\bibliographystyle{apalike}
\bibliography{references}

\end{document}